\documentclass[runningheads]{llncs}
\usepackage[T1]{fontenc}
\usepackage{amsmath, amssymb, graphicx}
\usepackage{adjustbox}
\usepackage{xcolor}
\usepackage{indentfirst}
\usepackage{verbatim}
\usepackage{todonotes}

\usepackage{shuffle}

\def\dsqcup{\sqcup\mathchoice{\mkern-7mu}{\mkern-7mu}{\mkern-3.2mu}{\mkern-3.8mu}\sqcup}%shuffle

\begin{document}
\title{Generalized Parikh Matrices For Tracking Subsequence Occurrences}
%
%\titlerunning{Abbreviated paper title}
% If the paper title is too long for the running head, you can set
% an abbreviated paper title here
%
\author{Szil\'ard Zsolt Fazekas \and
Huang Xinhao}
\authorrunning{S.Z. Fazekas \and H. Xinhao}
% First names are abbreviated in the running head.
% If there are more than two authors, 'et al.' is used.
%
\institute{Akita University, Akita, Japan\\
\email{szilard.fazekas@ie.akita-u.ac.jp\quad\quad d8523004@s.akita-u.ac.jp}
% \and
%Akita University, Akita, Japan \\
%\email{d8523004@s.akita-u.ac.jp}\\
 }
\maketitle              % typeset the header of the contribution
%
\begin{comment}

\begin{abstract}
We introduce and study an extended Parikh matrix mapping which is based on a special type of matrices. The extended Parikh matrices give much more information about a word than just the number of occurrences of scattered factors. First we give an extension of the Parikh matrix mapping called Parikh factor matrix mapping and extend it to Parikh sequence matrix mapping which is a morphism from the number of occurrences of generalized subsequences to matrices. In particular it could be an important tool in the theory of formal languages.\\
Next, we introduce a closely related notion named generalized subword histories and establish an interesting interconnection between extended Parikh matrices and classical Parikh matrices. we also prove that each minor of a sharping of Parikh sequence matrices has nonnegative determinant. Last, we mention the ambiguity problem under extended Parikh matrices.

\keywords{Parikh matrix  \and Scattered subsequences \and Subword history.}
\end{abstract}
\end{comment}

\begin{abstract}

We introduce and study a generalized Parikh matrix mapping based on tracking the occurrence counts of special types of subsequences. These matrices retain more information about a word than the original Parikh matrix mapping while preserving the homomorphic property. We build the generalization by first introducing the Parikh factor matrix mapping and extend it to the Parikh sequence matrix mapping. We establish an interesting connection between the generalized Parikh matrices and the original ones and use it to prove that certain important minors of a Parikh sequence matrix have nonnegative determinant. Finally, we generalize the concept of subword histories and show that each generalized subword history is equivalent to a linear one.

\end{abstract}

\section{Introduction}
The study of scattered subwords in combinatorics on words has been a topic of interest for nearly a century. Scattered subwords, also known as scattered factors, subsequences, or simply subwords, are defined as follows: for words $u$ and $v$, $u$ is a scattered subword of $v$ if there exist words $x_1, \ldots, x_n$ and $y_0, \ldots, y_n$, some of which may be empty, such that $u = x_1 \cdots x_n$ and $v = y_0 x_1 y_1 \cdots x_n y_n$. For instance, if $u = ab$ and $v = acbab$, then $u$ is a scattered subword of $v$. 

A fundamental result in this area is Higman's theorem (1952), which states that any subset of words over a finite alphabet that are not pairwise comparable for the subword ordering is finite, thus establishing that subword ordering is a well partial ordering. Higman's theorem has been rediscovered multiple times, often also credited to Haines (1969). Since then, the concept of subwords has garnered increasing attention from researchers.

One notable recent application of the theory of subwords is in the study of downward closures. The downward closure of a language is the set of all scattered subwords of its members. Abstracting languages by their downward closure as proxy can make equivalence and inclusion decidable/efficient when direct comparisons are undecidable or computationally hard~\cite{Zetzsche16,Zetzsche18}.

Another topic that garnered a lot of interest in the last few years is the algorithmics of the so called Simon congruence, that compares the set of all subwords of given words. A recent breakthrough result is that this congruence is efficiently decidable~\cite{GawryKKMS21}, which was a starting point to investigating the question in the context of pattern matching~\cite{KIM2024114478,FleishmannSKMNSW23}, and subword analysis problems~\cite{SungminHSS23}.

Our research focuses on the direction that began with the Parikh vector mapping, introduced by R. J. Parikh~\cite{Parikh66}, which has become a significant concept in the theory of formal languages. The Parikh vector of a word records the number of occurrences of each letter in the word. Although mapping a word to a vector in this way loses much information, there have been several refinements of the mapping to address that, perhaps most notably, by Mateescu et al., who introduced the so called Parikh matrices~\cite{MateescuSSY01}. Parikh matrices count the occurrences of scattered subwords instead of letters, with the Parikh vector appearing on the second diagonal of the matrix. These matrices have proven to be a powerful tool in language theory and combinatorics and generated a vast literature.

A Parikh vector usually represents multiple words, leading to an ambiguous representation. Salomaa investigated conditions for uniquely determining a word from its Parikh matrix~\cite{Salomaa05}. Although Parikh matrices can reduce the ambiguity inherent in the Parikh type of mappings, they cannot eliminate it entirely. Extensive work has been done on sharpening these mappings. The $q$-Parikh matrices, introduced in \cite{RenardRW202401,RenardRW202402}, extend the definition of binomial coefficients to q-deformations, recording occurrences and positions of each letter for richer information. Serbanuta considered the original Parikh matrices, which only track letters following an ordered alphabet, and proposed an extended Parikh matrix mapping where the matrix rank is determined by the word length rather than the alphabet size \cite{Serbanuta04}. Dick et al. introduced paired mappings by so called $\mathbb{P}$-Parikh matrices and $\mathbb{L}$-Parikh matrices~\cite{DickHMR20}.

Since the Parikh matrix mapping is not injective, Atanasiu examined properties of binary words with the same Parikh matrix, referring to these words as amiable \cite{AtanasiuAP07}. He characterized conditions under which equivalence classes of amiable words contain multiple elements and extended some of the results to larger alphabets \cite{AtanasiuAP08}. A long awaited recent development~\cite{HahnCH23} gave conditions over ternary alphabets for Words to be amiable, using a finite set of rewriting rules to obtain equivalence classes.

In this paper, we extend Parikh matrices to be able track subsequences in which we can require that some of the letters are matched to consecutive positions. This is a significantly stronger formalism than simple subwords, although not quite as flexible as others, such as subsequences with gap constraints. However, based on our results it seems that expanding further in that direction may be possible by minor adaptations of our results. In Section 3, we first develop simple matrices that count the number of occurrences of factors. Building on that, in Section 4 we introduce Parikh sequence matrices. We show that both mappings are homomorphisms and examine the determinants of certain minors of special interest. We then define generalized subword histories and we show how to reduce generalized to linear forms, laying the foundations of generalizing a vast collection of results about Parikh matrices to these new, sharper mappings.

\section{Preliminaries}
%  Let $u,v$ be words over an alphabet $\Sigma$. The word $u$ is a factor of $v$ if there exists words $s,t\in \Sigma^*$ such that $v= sut$.
%In this paper, we follow the fundamentals in \cite{MateescuSSY01}, and first introduce some terms about words. 
We will denote by $\mathbb{N}$ the set of natural numbers starting from $1$ and $\mathbb{N}_0=\mathbb{N}\cup \{0\}$. For all $i,j\in \mathbb{N}_0$ with $i\leq j$, we denote the closed interval between $i$ and $j$ by $[i,j]=\{i,\dots,j\}$. For succinctness we also use $[n]=[0,n]$.
An ordered alphabet is an alphabet $\Sigma=\{a_1,a_2,...,a_k\}$ with a linear order $"<"$ on it. If we have $a_1<a_2<...<a_k$, then we use the notation $\Sigma=\{a_1<a_2<...<a_k\}$. 

For a word $u=a_1\cdots a_n$, where $a_i\in \Sigma$ for all $1\leq i\leq n$, the length of $u$ is denoted $|u|$ and it is equal to $n$. The word of length $0$, having no letters, is called the \emph{empty word} and is denoted by $\varepsilon$.
Let $u,v$ be words over an alphabet $\Sigma$. The word $u$ is a \textit{factor} of $v$ if there exists words $s,t\in \Sigma^*$ such that $v= sut$.
The word $u$ is a \textit{(scattered) subword} of $v$ if there exists letters $a_1,...,a_n$ and words $t_0,...,t_n\in \Sigma^*$ such that $u=a_1\cdots a_n$ and $v=t_0a_1t_1 \cdots a_nt_n$. Now let $u=a_1\cdots a_n$. For some $i\in [1,n]$, the $i$th letter of $u$ is denoted by $u[i]$, that is $u[i]=a_i$. For all $i,j\in [1,n]$ with $i\leq j$, the factor of $u$ between positions $i$ and $j$ is denoted $u[i,j]=a_i\cdots a_j$. An \emph{occurrence as a factor} of $u$ in $v$ is a position $i\in [1,|v|]$ such that $v[i,i+|u|-1]=u$. The number of different occurrences of $u$ in $v$ as a subword is denoted $|v|_u$. An \textit{occurrence as a subword} of $u$ in $v$ is a tuple $(i_1,\dots,i_n)$ of increasing elements of $[1,|v|]$ such that $v[i_j]=a_j$ for all $j\in [1,|u|]$. We denote the number of distinct occurrences of a nonempty word $u$ as a scattered subword in $v$ by writing ${v\choose u}$. For example, ${aab\choose a}=2$, ${aaabb\choose ab}=6$. To simplify arguments, sometimes we will refer to an occurrence and mean the actual positions where the letters of the factor/subword is matched in the containing word.

The starting point of our study are the so called Parikh matrices, introduced by Mateescu et al., which are square matrices where the entries track subword occurrences, and are named after Parikh, because they represent a sharpening of the well-known Parikh vector mapping from words over $\Sigma$ to $\mathbb{N}_0^{|\Sigma|}$. In what follows, $\mathcal{M}_{n}$ denotes the monoid of $n\times n$ integer matrices with multiplication.

%\todo[inline]{For definitions and result cited from other works, you need to put a \cite also in the definition/theorem as below}
\begin{definition}[Parikh matrix mapping~\cite{MateescuSSY01}]
    Let $\Sigma=\{a_1<a_2<...<a_k\}$ be an arbitrary alphabet, where $k\geq 1$. The Parikh matrix mapping, denoted $$\Psi_{M_k}:\Sigma^*\to \mathcal{M}_{k+1},$$
    defined by the condition: if $\Psi_{M_k}(a_q)=(m_{i,j})_{1 \leq i,j \leq (k+1)}$, then for each $1\leq i \leq (k+1)$, $m_{i,i}=1$, $m_{q,q+1}=1$, all other elements of the matrix $\Psi_{M_k}(a_q)$ being 0.
\end{definition}

Many general properties of the Parikh matrix mapping were proved in the same paper they were introduced. The most fundamental is the following.

\begin{theorem}[\cite{MateescuSSY01}, Theorem 2.1]
    Let $\Sigma=\{a_1<a_2<...<a_k\}$ be an ordered alphabet, where $k\geq 1$, and assume that $w\in \Sigma^*$. The matrix $\Psi_{M_k}(w)=(m_{i,j})_{1\leq i,j\leq (k+1)}$, has the following properties:
    \begin{enumerate}
    \item  
    $m_{i,j}=0$, for all $1\leq j<i\leq (k+1)$; 
    \item  
    $m_{i,i}=1$, for all $1\leq i\leq (k+1)$; 
    \item  
    $m_{i,j+1}={w \choose a_{i}\cdots a_j}$, for all $1\leq i\leq j\leq k$.
\end{enumerate}
\end{theorem}

From here we get that the Parikh matrix mapping is a homomorphism and the entries of the matrix track the number of occurrences of factors of $a_1\cdots a_n$ as subwords of the argument of the mapping, as the following example illustrates.
\begin{example}
    Consider the word $w=abcb$ over alphabet $\Sigma=\{a<b<c\}$. Since the Parikh matrix mapping is a homomorphism, we can write
    \begin{alignat}{2}
    \Psi_{M_k}(w)&=\Psi_{M_k}(a)\Psi_{M_k}(b)\Psi_{M_k}(c)\Psi_{M_k}(b) \notag \\
    &=\begin{bmatrix}
        1 & 1 & 0 & 0\\
        0 & 1 & 0 & 0\\
        0 & 0 & 1 & 0\\
        0 & 0 & 0 & 1 
        \end{bmatrix} \begin{bmatrix}
        1 & 0 & 0 & 0\\
        0 & 1 & 1 & 0\\
        0 & 0 & 1 & 0\\
        0 & 0 & 0 & 1 
        \end{bmatrix} \begin{bmatrix}
        1 & 0 & 0 & 0\\
        0 & 1 & 0 & 0\\
        0 & 0 & 1 & 1\\
        0 & 0 & 0 & 1 
        \end{bmatrix} \begin{bmatrix}
        1 & 0 & 0 & 0\\
        0 & 1 & 1 & 0\\
        0 & 0 & 1 & 0\\
        0 & 0 & 0 & 1 
        \end{bmatrix}
        \notag\\
        &=\begin{bmatrix}
        1 & 1 & 2 & 1\\
        0 & 1 & 2 & 1\\
        0 & 0 & 1 & 1\\
        0 & 0 & 0 & 1 
        \end{bmatrix} 
        =\begin{bmatrix}
            1 & {w\choose a} & {w\choose ab} & {w\choose abc}\\[0.2cm]
            0 & 1 & {w\choose b} & {w\choose bc}\\[0.2cm]
            0 & 0 & 1 & {w\choose c}\\[0.2cm]
            0 & 0 & 0 & 1 
        \end{bmatrix}\notag
    \end{alignat}
\end{example}

The first extension of the Parikh matrix mapping relaxed the condition on the kind of subwords that can be tracked in these matrices~\cite{Serbanuta04}. The formalism is almost identical, but for the extended Parikh matrix mapping induced by an \textit{arbitrary} word $v=b_1\cdots b_k$, the entries above the main diagonal become $m_{i,j+1}={w\choose b_i\cdots b_j}$ for all $1\leq i\leq j\leq k$. The original Parikh matrix is the special case when $v$ is just the concatenation of the letters of the alphabet.

\section{Parikh factor matrices}
Parikh matrices only count the occurrences of \textit{scattered subwords}. In order to generalize this mapping for tracking more complex subsequences, as a first step we introduce Parikh-style matrices capable of counting occurrences of \textit{factors}.

In general, to be able to construct a homomorphic mapping that computes the number of occurrences as factors of some word $u$ in a word $v_1v_2$, trivially the mapping must preserve information about $|v_1|_u$, $|v_2|_u$, but also about suffixes of $v_1$ that are prefixes of $u$ and prefixes of $v_2$ that are suffixes of $u$, since $|v_1v_2|_u=|v_1|_u+|v_2|_u+ D$, where $D$ is the number of decompositions $u=u'u''$ such that $u'$ is a nonempty suffix of $v_1$ and $u''$ is a nonempty prefix of $v_2$. To illustrate the idea on a simple case, we begin by constructing the matrices to track the number of occurrences of $ab$ as a factor. In this case the necessary information about each word is the number of $ab$ factors in it, whether they finish with $a$ and whether they start with $b$. We introduce the following mappings, where $u,v,w\in \Sigma^*$:
\begin{align}
s^w_u &=\begin{cases}
  1 & \text{ if } w \text{ starts with } u\text{ } (w=uv)\\
  0 & \text{ otherwise } \end{cases}\notag\qquad
c^w_v =\begin{cases}
  1 & \text{ if } w = v\\
  0 & \text{ otherwise } \end{cases} %\text{ factor }u,v\in \{a,b\}^l, l<n. 
  \notag\\
e^w_v &=\begin{cases}
  1 & \text{ if } w \text{ ends with } v\text{ } (w=uv)\\
  0 & \text{ otherwise } \end{cases}\notag
 \end{align} \\[0.4cm]
and for a more uniform (and more compact) notation use $f^w_u$ to indicate the number of factors $u$ in word $w$, instead of the previously introduced and more usual $|w|_u$. To be able to track the necessary values through matrix multiplication, the simplest way to arrange them in a matrix is
$$\Phi_M(w) =\begin{bmatrix}
 1\quad & e^w_a\quad & f^w_{ab}\\
 0\quad & c^w_\varepsilon\quad & s^w_b\\
 0\quad & 0\quad & 1
\end{bmatrix}$$ and for any two binary words $w_1,w_2$, it is straightforward that
$$\Phi_M(w_1)\Phi_M(w_2) =\begin{bmatrix}
 1 & e^{w_1}_a & f^{w_1}_{ab}\\
 0 & c^{w_1}_\varepsilon & s^{w_1}_b\\
 0 & 0 & 1
\end{bmatrix}\begin{bmatrix}
 1 & e^{w_2}_a & f^{w_2}_{ab}\\
 0 & c^{w_2}_\varepsilon & s^{w_2}_b\\
 0 & 0 & 1
\end{bmatrix}$$
$$
=\begin{bmatrix}
 1\quad & e^{w_2}_a+e^{w_1}_a\cdot c^{w_2}_\varepsilon\quad & f^{w_1}_{ab}+f^{w_2}_{ab}+e^{w_1}_a\cdot s^{w_2}_b\\
 0\quad & c^{w_1}_\varepsilon\cdot c^{w_2}_\varepsilon\quad & s^{w_1}_b+c^{w_1}_\varepsilon\cdot s^{w_2}_b\\
 0\quad & 0\quad & 1
\end{bmatrix}=\Phi_M(w_1w_2).
$$

Following the idea above we can generalize the matrix so that it could record the number of longer factors in a word, as follows.

\begin{definition}[Parikh factor matrix mapping]
    Let $\Sigma=\{a_1<\cdots< a_k\}$ be an alphabet and consider the concatenation of all the letters $\sigma=a_1a_2...a_k$, where $k\geq 1$. The Parikh factor matrix mapping, denoted $$\Phi_{M_\sigma}:\Sigma^*\to \mathcal{M}_{3(k-1)},$$
    is defined for any $w\in\Sigma^*$ as
    $$\Phi_{M_\sigma}(w) =\begin{bmatrix}
        I & E^w & F^w\\
        O & C^w & S^w\\
        O & O & I \end{bmatrix},$$
        where $I$ is the identity matrix, $O$ is the zero matrix and
$$F^w=\begin{bmatrix}
  f^w_{a_1a_2} & f^w_{a_1a_2a_3} & \dots  & f^w_{a_1...a_k}\\
  0      & f^w_{a_2a_3}    & \dots  & f^w_{a_2...a_k}\\
  \vdots & \vdots    & \ddots & \vdots \\
  0      & 0         & \dots  & f^w_{a_{k-1}a_k}
\end{bmatrix},$$ 
is a matrix tracking occurrences of factor of $a_1\dots a_k$ as factors of $w$. Finally, $S^w$ and $E^w$ include the necessary information about starting factors (prefixes) and end factors (suffixes):
$$S^w=\begin{bmatrix}
  s^w_{a_2}    & s^w_{a_2a_3} & \dots  & s^w_{a_2...a_k}\\
  0      & s^w_{a_3}    & \dots  & s^w_{a_3...a_k}\\
  \vdots & \vdots & \ddots & \vdots \\
  0      & 0      & \dots  & s^w_{a_k}
\end{bmatrix},\;
E^w=\begin{bmatrix}
  e^w_{a_1}    & e^w_{a_1a_2} & \dots  & e^w_{a_1...a_{k-1}}\\
  0      & e^w_{a_2}    & \dots  & e^w_{a_2...a_{k-1}}\\
  \vdots & \vdots & \ddots & \vdots \\
  0      & 0      & \dots  & e^w_{a_{k-1}}
\end{bmatrix},$$
while matrix $C$ makes sure that the mapping works correctly for short words (shorter than $k$), too,
$$C^w=\begin{bmatrix}
  c^w_\varepsilon & c^w_{a_2} & \dots  & c^w_{a_2...a_{k-1}}\\
  0      & c^w_\varepsilon    & \dots  & c^w_{a_3...a_{k-1}}\\
  \vdots & \vdots & \ddots & \vdots \\
  0      & 0      & \dots  & c^w_\varepsilon
\end{bmatrix}.$$
\end{definition}

Next we argue that the mapping indeed works for tracking the occurrences when concatenating the containing words, that is, the mapping is homomorphic.

\begin{theorem}
    The Parikh factor matrix mapping $\Phi_{M_\sigma}:\Sigma^*\to \mathcal{M}_{3(k-1)}$ is a homomorphism.
\end{theorem}
\begin{proof}
    For all words $w_1, w_2\in\Sigma^*$, $\sigma=a_1a_2...a_k$, where $k\geq 1$ and $a_1,a_2,...,a_k\in \Sigma$, the following holds.

$$\begin{adjustbox}{max width=\linewidth}
$\Phi_{M_\sigma}(w_1)\Phi_{M_\sigma}(w_2)=\begin{bmatrix}
  I & E^{w_1} & F^{w_1}\\
  O & C^{w_1} & S^{w_1}\\
  O & O & I
\end{bmatrix}
\begin{bmatrix}
  I & E^{w_2} & F^{w_2}\\
  O & C^{w_2} & S^{w_2}\\
  O & O & I
\end{bmatrix}=\begin{bmatrix}
  I\quad & E^{w_2}+E^{w_1}\cdot C^{w_2}\quad & F^{w_1}+F^{w_2}+E^{w_1}\cdot S^{w_2}\\
  O\quad & C^{w_1}\cdot C^{w_2}\quad & S^{w_1}+C^{w_1}\cdot S^{w_2}\\
  O\quad & O\quad & I
\end{bmatrix}$\end{adjustbox}$$

$$\Phi_{M_\sigma}(w_1w_2)=\begin{bmatrix}
  I\quad & E^{w_1w_2}\quad & F^{w_1w_2}\\
  O\quad & C^{w_1w_2}\quad & S^{w_1w_2}\\
  O\quad & O\quad & I
\end{bmatrix}$$
To see the validity of the statement, we compare the matrices component-wise, with the equalities directly following from the definitions of $s$, $c$, $e$ and $f$.
\begin{alignat}{2}
 E^{w_1w_2}&=\begin{bmatrix}
  e^{w_1w_2}_{a_1}    & e^{w_1w_2}_{a_1a_2} & \dots  & e^{w_1w_2}_{a_1...a_{k-1}}\\
  0      & e^{w_1w_2}_{a_2}    & \dots  & e^{w_1w_2}_{a_2...a_{k-1}}\\
  \vdots & \vdots & \ddots & \vdots \\
  0      & 0      & \dots  & e^{w_1w_2}_{a_{k-1}}
\end{bmatrix}\notag\\
&=\begin{bmatrix}
  e^{w_2}_{a_1}    & e^{w_2}_{a_1a_2} & \dots  & e^{w_2}_{a_1...a_{k-1}}\\
  0      & e^{w_2}_{a_2}    & \dots  & e^{w_2}_{a_2...a_{k-1}}\\
  \vdots & \vdots & \ddots & \vdots \\
  0      & 0      & \dots  & e^{w_2}_{a_{k-1}}
\end{bmatrix}\notag\\
&+\begin{bmatrix}
  e^{w_1}_{a_1}c^{w_2}_\varepsilon    & e^{w_1}_{a_1}c^{w_2}_{a_2}+e^{w_1}_{a_1a_2}c^{w_2}_\varepsilon & \dots  & e^{w_1}_{a_1}c^{w_2}_{a_2...a_{k-1}}+...+e^{w_1}_{a_1...a_{k-1}}c^{w_2}_\varepsilon\\
  0      & e^{w_1}_{a_2}c^{w_2}_\varepsilon    & \dots  & e^{w_1}_{a_2}c^{w_2}_{a_3...a_{k-1}}+...+e^{w_1}_{a_2...a_{k-1}}c^{w_2}_\varepsilon\\
  \vdots & \vdots & \ddots & \vdots \\
  0      & 0      & \dots  & e^{w_1}_{a_{k-1}}c^{w_2}_\varepsilon
\end{bmatrix} \notag\\
&=E^{w_2}+E^{w_1}\cdot C^{w_2} \notag
\end{alignat}
\begin{alignat}{2}
 S^{w_1w_2}&=\begin{bmatrix}
  s^{w_1w_2}_{a_2}    & s^{w_1w_2}_{a_2a_3} & \dots  & s^{w_1w_2}_{a_2...a_k}\\
  0      & s^{w_1w_2}_{a_3}    & \dots  & s^{w_1w_2}_{a_3...a_k}\\
  \vdots & \vdots & \ddots & \vdots \\
  0      & 0      & \dots  & s^{w_1w_2}_{a_k}
\end{bmatrix}\notag\\
&=\begin{bmatrix}
  s^{w_1}_{a_2}    & s^{w_1}_{a_2a_3} & \dots  & s^{w_1}_{a_2...a_k}\\
  0      & s^{w_1}_{a_3}    & \dots  & s^{w_1}_{a_3...a_k}\\
  \vdots & \vdots & \ddots & \vdots \\
  0      & 0      & \dots  & s^{w_1}_{a_k}
\end{bmatrix}\notag\\
&+\begin{bmatrix}
  c^{w_1}_\varepsilon s^{w_2}_{a_2}   & c^{w_1}_\varepsilon s^{w_2}_{a_2a_3}+c^{w_1}_{a_2}s^{w_2}_{a_3} & \dots  & c^{w_1}_\varepsilon s^{w_2}_{a_2...a_k}+...+c^{w_1}_{a_2...a_{k-1}}s^{w_2}_{a_k}\\
  0      & c^{w_1}_\varepsilon s^{w_2}_{a_3}+c^{w_1}_{a_2}s^{w_2}_{a_3} & \dots  & c^{w_1}_\varepsilon s^{w_2}_{a_3...a_k}+...+c^{w_1}_{a_3...a_{k-1}}s^{w_2}_{a_k}\\
  \vdots & \vdots & \ddots & \vdots \\
  0      & 0      & \dots  & c^{w_1}_\varepsilon s^{w_2}_{a_k}
\end{bmatrix} \notag\\
&=S^{w_1}+C^{w_1}\cdot S^{w_2}  \notag
\end{alignat}
\begin{alignat}{2}
 C^{w_1w_2}&=\begin{bmatrix}
  c^{w_1w_2}_\varepsilon & c^{w_1w_2}_{a_2} & \dots  & c^{w_1w_2}_{a_2...a_{k-1}}\\
  0      & c^{w_1w_2}_\varepsilon    & \dots  & c^{w_1w_2}_{a_3...a_{k-1}}\\
  \vdots & \vdots & \ddots & \vdots \\
  0      & 0      & \dots  & c^{w_1w_2}_\varepsilon
\end{bmatrix}\notag\\
&=\begin{bmatrix}
  c^{w_1}_\varepsilon & c^{w_1}_{a_2} & \dots  & c^{w_1}_{a_2...a_{k-1}}\\
  0      & c^{w_1}_\varepsilon    & \dots  & c^{w_1}_{a_3...a_{k-1}}\\
  \vdots & \vdots & \ddots & \vdots \\
  0      & 0      & \dots  & c^{w_1}_\varepsilon
\end{bmatrix}
\cdot\begin{bmatrix}
  c^{w_2}_\varepsilon & c^{w_2}_{a_2} & \dots  & c^{w_2}_{a_2...a_{k-1}}\\
  0      & c^{w_2}_\varepsilon    & \dots  & c^{w_2}_{a_3...a_{k-1}}\\
  \vdots & \vdots & \ddots & \vdots \\
  0      & 0      & \dots  & c^{w_2}_\varepsilon
\end{bmatrix} \notag\\
&=C^{w_1}\cdot C^{w_2} \notag
\end{alignat}

For $F^{w_1w_2}$, we get

$F^{w_1w_2}=\begin{bmatrix}
  f^{w_1w_2}_{a_1a_2} & f^{w_1w_2}_{a_1a_2a_3} & \dots  & f^{w_1w_2}_{a_1...a_k}\\
  0      & f^{w_1w_2}_{a_2a_3}    & \dots  & f^{w_1w_2}_{a_2...a_k}\\
  \vdots & \vdots    & \ddots & \vdots \\
  0      & 0         & \dots  & f^{w_1w_2}_{a_{k-1}a_k}
\end{bmatrix}$

\begin{adjustbox}{max width=\linewidth}
$=\begin{bmatrix}
  f^{w_1}_{a_1a_2}+f^{w_2}_{a_1a_2}+e^{w_1}_{a_1}\cdot s^{w_2}_{a_2} & f^{w_1}_{a_1a_2a_3}+f^{w_2}_{a_1a_2a_3}+e^{w_1}_{a_1}\cdot s^{w_2}_{a_2a_3}+e^{w_1}_{a_1a_2}\cdot s^{w_2}_{a_3} & \dots  & f^{w_1}_{a_1...a_k}+f^{w_2}_{a_1...a_k}+e^{w_1}_{a_1}\cdot s^{w_2}_{a_2a_3...a_k}+...+e^{w_1}_{a_1...a_{k-1}}\cdot s^{w_2}_{a_k}\\
  0      & f^{w_1}_{a_2a_3}+f^{w_2}_{a_2a_3}+e^{w_1}_{a_2}\cdot s^{w_2}_{a_3}    & \dots  & f^{w_1}_{a_2...a_k}+f^{w_2}_{a_2...a_k}+e^{w_1}_{a_2}\cdot s^{w_2}_{a_3...a_k}+...+e^{w_1}_{a_2...a_{k-1}}\cdot s^{w_3}_{a_k}\\
  \vdots & \vdots    & \ddots & \vdots \\
  0      & 0         & \dots  & f^{w_1}_{a_{k-1}a_k}+f^{w_2}_{a_{k-1}a_k}+e^{w_1}_{a_{k-1}}\cdot s^{w_2}_{a_k}
\end{bmatrix}$
\end{adjustbox}
$$\begin{adjustbox}{max width=\linewidth}
$=\begin{bmatrix}
 f^{w_1}_{a_1a_2} & f^{w_1}_{a_1a_2a_3} & \dots & f^{w_1}_{a_1...a_k} \\
 0 & f^{w_1}_{a_2a_3} & \dots & f^{w_1}_{a_2...a_k} \\
 \vdots & \vdots    & \ddots & \vdots \\
   0      & 0         & \dots  & f^{w_1}_{a_{k-1}a_k}
\end{bmatrix}+$
$\begin{bmatrix}
  f^{w_2}_{a_1a_2} & f^{w_2}_{a_1a_2a_3} & \dots  & f^{w_2}_{a_1...a_k}\\
  0 & f^{w_2}_{a_2a_3} & \dots & f^{w_2}_{a_2...a_k}\\
  \vdots & \vdots    & \ddots & \vdots \\
  0      & 0         & \dots  & f^{w_2}_{a_{k-1}a_k}
\end{bmatrix}+$
$\begin{bmatrix}
  e^{w_1}_{a_1} & e^{w_1}_{a_1a_2} & \dots  & e^{w_1}_{a_1...a_{k-1}}\\
  0      & e^{w_1}_{a_2}    & \dots  & e^{w_1}_{a_2...a_{k-1}}\\
  \vdots & \vdots    & \ddots & \vdots \\
  0      & 0         & \dots  & e^{w_1}_{a_{k-1}}
\end{bmatrix}\cdot$
$\begin{bmatrix}
  s^{w_2}_{a_2} & s^{w_2}_{a_2a_3} & \dots  & s^{w_2}_{a_2a_3...a_k}\\
  0  & s^{w_2}_{a_3} & \dots & s^{w_2}_{a_3...a_k}\\
  \vdots & \vdots    & \ddots & \vdots \\
  0      & 0         & \dots  & s^{w_2}_{a_k}
\end{bmatrix}$
\end{adjustbox}$$
$=F^{w_1}+F^{w_2}+E^{w_1}\cdot S^{w_2} $\\
So, $\Phi_{M_\sigma}(w_1)\Phi_{M_\sigma}(w_2)=\Phi_{M_\sigma}(w_1w_2)$, for all words $w_1,w_2\in\Sigma^*$, and the Parikh factor matrix mapping is a morphism.

\qed\end{proof}

\begin{example}
Let $\Sigma=\{a,b,c\}$, $w_1,w_2\in \Sigma^*$, $w_1=b,w_2=cabc$, $\sigma=abc$,
$$\begin{array}{l} 
    \Phi_{M_\sigma}(w_1)\Phi_{M_\sigma}(w_2)=\begin{bmatrix}
      I & E^{w_1} & F^{w_1}\\
      O & C^{w_1} & S^{w_1}\\
      O & O & I
    \end{bmatrix}\begin{bmatrix}
      I & E^{w_2} & F^{w_2}\\
      O & C^{w_2} & S^{w_2}\\
      O & O & I
    \end{bmatrix} \\
    
    =\begin{bmatrix}
      I & \begin{bmatrix}
      e^{w_1}_a & e^{w_1}_{ab}   \\
      0      & e^{w_1}_b      \\
    \end{bmatrix} & \begin{bmatrix}
      f^{w_1}_{ab} & f^{w_1}_{abc}  \\
      0      & f^{w_1}_{bc}     \\
    \end{bmatrix}\\
      O & \begin{bmatrix}
     c^{w_1}_\varepsilon & c^{w_1}_b  \\
     0 & c^{w_1}_\varepsilon  \\
     \end{bmatrix} & \begin{bmatrix}
      s^{w_1}_b & s^{w_1}_{bc} \\
      0      & s^{w_1}_c     \\
    \end{bmatrix}\\
      O & O & I
    \end{bmatrix}
    \begin{bmatrix}
      I & \begin{bmatrix}
      e^{w_2}_a & e^{w_2}_{ab}  \\
      0      & e^{w_2}_b    \\
    \end{bmatrix} & \begin{bmatrix}
      f^{w_2}_{ab} & f^{w_2}_{abc}   \\
      0      & f^{w_2}_{bc}      \\
    \end{bmatrix}\\
      O & \begin{bmatrix}
     c^{w_2}_\varepsilon & c^{w_2}_b\\
     0 & c^{w_2}_\varepsilon \\
     \end{bmatrix} & \begin{bmatrix}
      s^{w_2}_b & s^{w_2}_{bc}   \\
      0      & s^{w_2}_c      \\
    \end{bmatrix}\\
      O & O & I
    \end{bmatrix}
    \\

=\begin{bmatrix}
  I & \begin{bmatrix}
  0 & 0   \\
  0      & 1    \\
\end{bmatrix} & \begin{bmatrix}
  0 & 0 \\
  0      & 0   \\
\end{bmatrix}\\
  O & \begin{bmatrix}
 0 & 1 \\
 0 & 0  \\
 \end{bmatrix} & \begin{bmatrix}
  1 & 0 \\
  0      & 0    \\
\end{bmatrix}\\
  O & O & I
\end{bmatrix}
\begin{bmatrix}
  I & \begin{bmatrix}
  0 & 0  \\
  0      & 0   \\
\end{bmatrix} & \begin{bmatrix}
  1 & 1   \\
  0      & 1     \\
\end{bmatrix}\\
  O & \begin{bmatrix}
 0 & 0\\
 0 & 0 \\
 \end{bmatrix} & \begin{bmatrix}
  0 & 0   \\
  0      & 1     \\
\end{bmatrix}\\
  O & O & I
\end{bmatrix} =\begin{bmatrix}
  I & \begin{bmatrix}
  0 & 0  \\
  0      & 0    \\
\end{bmatrix} & \begin{bmatrix}
 1  & 1   \\
  0      & 2    \\
\end{bmatrix}\\
  O & \begin{bmatrix}
 0 & 0\\
 0 & 0 \\
 \end{bmatrix} & \begin{bmatrix}
  1 & 1   \\
  0      & 0    \\
\end{bmatrix}\\
  O & O & I
\end{bmatrix}
=\Phi_{M_\sigma}(w_1w_2)
\end{array}$$

\end{example}

\section{Parikh sequence matrices}
Now that we have already constructed the Parikh factor matrices and showed that the mapping is a homomorphism, we are ready to generalize the Parikh matrix mapping to track more flexibly defined subsequences. 
We will use the operation $\bullet$ for combining factors $q_1,q_2$ over $\Sigma$ to generalize the notion of subwords to \emph{generalized subsequences}. These are still subwords, but with conditions attached to them on certain letters having to occur right next to each other in the containing word. 

\begin{definition}[Generalized subsequences]\label{def:gs}
    Let $q_1,q_2,...,q_n,v$ be nonempty words over $\Sigma$. The sequence $Q=q_1\bullet q_2\bullet ...\bullet q_n$ is a \emph{generalized subsequence} of $v$ if there exists words $t_0,t_1,...,t_n\in \Sigma^*$ such that $v=t_0 q_1 t_1 ... q_n t_n$. An occurrence of $Q$ in $v$ is a tuple $(i_1,\dots, i_n)$ of increasing positions of $v$ such that for each $j\in [1,n-1]$ we have $i_{j+1}-i_j\geq |q_j|$ and $v[i_j,i_j+|q_j|-1]=q_j$. The number of distinct occurrences of $Q$ in $v$ will be denoted as $|v|_Q$.
\end{definition}

\begin{example}
    The $\bullet$ notation intuitively represents a gap (possibly of length $0$), when matching a generalized subsequence in the word containing it. For instance $|aabb|_{ab\bullet b}=1$, but $|aabb|_{a\bullet bb}=2$. 
\end{example}

Note that when $q_1,\dots,q_n$ are all single letters, then $q_1\bullet \cdots \bullet q_n$ being a generalized subsequence of $v$ is equivalent to saying that $q_1\cdots q_n$ is a subword of $v$, and $|w|_{q_1\bullet \cdots \bullet q_n}={w\choose q_1\cdots q_n}$. For convenience we set $q\bullet \varepsilon=\varepsilon\bullet q=q$ for any $q\in \Sigma^*$ and we adopt the convention that $|w|_\varepsilon=1$ for any $w\in \Sigma^*$.

Such subsequences have been considered in many contexts and under various names, such as gapped sequences~\cite{EdwardsS04,AndreattaN16} in bioinformatics or subsequences with gap constraints~\cite{DayKMS22}, a notion that is even more general than what we suggest here. The number of occurrences of those subsequences can be thought of as lossy representations of the containing word and are natural mathematical models for situations where one has to deal with
input strings with missing or erroneous symbols in DNA sequences or digital signals.
%Parikh sequence matrix will keep track the number of occurrences of a sequence $f_1,f_2,\dots,f_k$, where $f_1,f_2,\dots,f_k$ are factors of word $w$ and factors do not have to follow right each others.
\subsection{Parikh sequence matrices on arbitrary words}

Similarly to the case of Parikh factor matrices, we use the following mappings for arbitrary words $u,v,w\in \Sigma^*$:
\begin{align}
s^w_{u\bullet v} &=\begin{cases}
  f^w_v & {\text{ if } w \text{ starts with } u\text{ or } u \text{ is the empty word.}}\\
  0 & {\text{ otherwise } }\end{cases} \notag \\
e^w_{u\bullet v} &=\begin{cases}
  f^w_u & {\text{ if } w \text{ ends with } v\text{ or }v \text{ is the empty word. }}\\
  0 &{ \text{ otherwise } }\end{cases} \notag \\
  c^w_{u\bullet v} &=\begin{cases}
  1 & \text{ if } w = uxv \text{ for some } x\in\Sigma^*\\
  0 & \text{ otherwise } \end{cases} \notag 
  \end{align}\\
  and write $f^w_{u\bullet v}$ to indicate the number of distinct factors of the form $uxv$ in $w$. For a sequence $Q=q_1\bullet q_2 \bullet \dots \bullet q_n$, we define $|Q|=\sum_{i=1}^n |q_i|$. In what follows, $Q[i:j]$ denotes the subsequence from $i$th letter to $j$th in $Q$, where the $\bullet$ operations do not count as taking up positions, so $ab\bullet cd [2:3]=b\bullet c$. 

\begin{definition}[Parikh sequence matrix mapping]
Let $\Sigma$ be an arbitrary alphabet and a generalized subsequence $Q=q_1\bullet q_2 \bullet \dots \bullet q_x$ where $x\geq 1$. The Parikh sequence matrix mapping, denoted $\Xi_{Q}$, is the morphism:
$$\Xi_{Q}:\Sigma^*\to \mathcal{M}_{3(|Q|-1)},$$
defined by
$$\Xi_{Q}(w)=\begin{bmatrix}            
  I & E^{w} & F^{w}\\
  O & C^{w} & S^{w}\\
  O & O & I
\end{bmatrix}$$

where 
$$E^w=\begin{bmatrix}
 e^{w}_{Q[1]} & \dots & e^{w}_{Q[1:l_1]\bullet} &\dots& e^{w}_{Q[1:l_2]\bullet} &\dots &  e^{w}_{Q[1:l_x-1]}\\
 \vdots & &\vdots & &\vdots & &\vdots\\
 0 &\dots& e^{w}_{Q[l_1] \bullet} &\dots& e^{w}_{Q[l_1:l_2]\bullet} &\dots& e^{w}_{Q[l_1:l_x-1]}\\
 \vdots & &\vdots & &\vdots & &\vdots\\
 0 &\dots& 0 &\dots& 0&\dots& e^{w}_{Q[l_x-2:l_x-1]}
 \end{bmatrix}$$$$
 F^w=\begin{bmatrix}
 f^{w}_{Q[1:2]} &\dots& f^{w}_{Q[1:l_1+1]}&\dots& f^{w}_{Q[1:l_2+1]} &\dots & f^{w}_{Q[1:l_x]} \\
 \vdots & &\vdots & &\vdots & &\vdots\\
 0 &\dots& f^{w}_{Q[l_1:l_1+1]}&\dots& f^{w}_{Q[l_1:l_2+1]} &\dots & f^{w}_{Q[l_1:l_x]} \\
 \vdots & &\vdots & &\vdots & &\vdots\\
 0 &\dots& 0 &\dots& 0&\dots& f^{w}_{Q[l_x-1:l_x]}
 \end{bmatrix}$$
 $$
 C^w=\begin{bmatrix}
 c^{w}_\varepsilon &\dots& s^{w}_{Q[2:l_1]} &\dots& s^{w}_{Q[2:l_2]}&\dots& c^{w}_{Q[2:l_x-1]} \\
 \vdots & &\vdots & &\vdots & &\vdots\\
 0 &\dots&  1 &\dots& e^{w}_{Q[l_1+1:l_2]} &\dots &e^{w}_{Q[l_1+1:l_x-1]}\\
 \vdots & &\vdots & &\vdots & &\vdots\\
 0 &\dots& 0 &\dots& 0 &\dots& c^{w}_\varepsilon 
 \end{bmatrix}$$$$
S^w=\begin{bmatrix}
 s^{w}_{Q[2]} &\dots& s^{w}_{Q[2:l_1+1]} &\dots& s^{w}_{Q[2:l_2+1]} &\dots& s^{w}_{Q[2:l_x]} \\
  \vdots & &\vdots & &\vdots & &\vdots\\
 0 &\dots& s^{w}_{\bullet Q[l_1+1]} &\dots& s^{w}_{\bullet Q[l_1+1:l_2+1]}&\dots&s^{w}_{\bullet Q[l_1+1:l_x]} \\
  \vdots & &\vdots & &\vdots & &\vdots\\
 0 &\dots &0 &\dots&0&\dots&  s^{w}_{Q[l_x]}
 \end{bmatrix}$$
 \end{definition}

\begin{theorem}
    For any generalized subsequence $Q$, the Parikh sequence matrix mapping $\Xi_{Q}:\Sigma^*\to \mathcal{M}_{3(|Q|-1)}$ is a homomorphism.
\end{theorem}
\begin{proof}
The argument is very similar to the case of Parikh matrices or Parikh factor matrices presented earlier, and we omit the formal description due to space constraints.

\qed\end{proof}

\begin{example}

Let $\Sigma=\{a,b,c\}$, $w_1,w_2\in \Sigma^*$, $w_1=babcab,w_2=cbcba,Q=ab\bullet c$.

$$\Xi_{M_{ab\bullet c}}(w_1)\Xi_{M_{ab\bullet c}}(w_2)=\begin{bmatrix}
  I & E^{w_1} & F^{w_1}\\
  O & C^{w_1} & S^{w_1}\\
  O & O & I
\end{bmatrix}\begin{bmatrix}
  I & E^{w_2} & F^{w_2}\\
  O & C^{w_2} & S^{w_2}\\
  O & O & I
\end{bmatrix} $$
$$\begin{adjustbox}{max width=\linewidth}
$
=\begin{bmatrix}
  I & \begin{bmatrix}
  e^{w_1}_a & e^{w_1}_{ab\bullet}   \\
  0      & e^{w_1}_{b\bullet}      \\
\end{bmatrix} & \begin{bmatrix}
  f^{w_1}_{ab} & f^{w_1}_{ab\bullet c}  \\
  0      & f^{w_1}_{b\bullet c}     \\
\end{bmatrix}\\
  O & \begin{bmatrix}
 c^{w_1}_\varepsilon & s^{w_1}_b  \\
 0 & 1 \\
 \end{bmatrix} & \begin{bmatrix}
  s^{w_1}_b & s^{w_1}_{b\bullet c} \\
  0      & s^{w_1}_{\bullet c }    \\
\end{bmatrix}\\
  O & O & I
\end{bmatrix}
\begin{bmatrix}
  I & \begin{bmatrix}
  e^{w_2}_a & e^{w_2}_{ab\bullet}  \\
  0      & e^{w_2}_{b \bullet }  \\
\end{bmatrix} & \begin{bmatrix}
  f^{w_2}_{ab} & f^{w_2}_{ab\bullet c}   \\
  0      & f^{w_2}_{b\bullet c}      \\
\end{bmatrix}\\
  O & \begin{bmatrix}
 c^{w_2}_\varepsilon & s^{w_2}_b\\
 0 & 1\\
 \end{bmatrix} & \begin{bmatrix}
  s^{w_2}_b & s^{w_2}_{b\bullet c}   \\
  0      & s^{w_2}_{\bullet c }     \\
\end{bmatrix}\\
  O & O & I
\end{bmatrix}$
\end{adjustbox}$$
$$
=\begin{bmatrix}
  I & \begin{bmatrix}
  0 & 2   \\
  0      & 3      \\
\end{bmatrix} & \begin{bmatrix}
  2 & 1  \\
  0      & 2    \\
\end{bmatrix}\\
  O & \begin{bmatrix}
 0 & 1  \\
 0 & 1  \\
 \end{bmatrix} & \begin{bmatrix}
  1 & 1 \\
  0      & 1    \\
\end{bmatrix}\\
  O & O & I
\end{bmatrix}
\begin{bmatrix}
  I & \begin{bmatrix}
  1 & 0  \\
  0      & 2    \\
\end{bmatrix} & \begin{bmatrix}
  0 & 0   \\
  0      & 1     \\
\end{bmatrix}\\
  O & \begin{bmatrix}
 0 & 0\\
 0 & 1 \\
 \end{bmatrix} & \begin{bmatrix}
  0 & 0   \\
  0      & 2     \\
\end{bmatrix}\\
  O & O & I
\end{bmatrix}=\begin{bmatrix}
  I & \begin{bmatrix}
  1 & 2  \\
  0      & 5    \\
\end{bmatrix} & \begin{bmatrix}
  2 & 5   \\
  0      & 9     \\
\end{bmatrix}\\
  O & \begin{bmatrix}
 0 & 1\\
 0 & 1 \\
 \end{bmatrix} & \begin{bmatrix}
  1 & 3   \\
  0      & 3     \\
\end{bmatrix}\\
  O & O & I
\end{bmatrix}=\Xi_{M_{ab\bullet c}}(w_1w_2)$$

\end{example}

\section{Determinants of minors}
In this section we look at the determinants of certain minors of the newly defined Parikh sequence matrices. We aim to obtain the generalization of the result that states that each minor of Parikh matrices has a nonnegative determinant~\cite{MateescuSSY01,MateescuSY04}. Our matrices contain two kinds of occurrence tracking entries, those terminated by a factor (auxiliary entries necessary for the mapping to be homomorphic) and those terminated by a bullet (the actual entries we are interested in tracking). The minors we consider here are made up of only the entries terminated by bullets. We will call those \emph{special minors}. Formally, for a word $w$ and a generalized subsequence $Q=q_1\bullet \cdots\bullet q_n$, the special minor of $\Xi_Q(w)$ is the one consisting of columns indexed $i_j$ and rows indexed $3(|Q|-1)+1-i_j$, for $j\in [1,n+1]$, where $i_1=1$, $i_j=|Q|-1+|q_1\cdots q_{j-1}|$, for each $j\in [2,n]$ and $i_{n+1}=3(|Q|-1)$, that is, the matrix
    $$M=\begin{bmatrix}
    1 & e^w_{q_1\bullet} & \cdots & f_{q_1\bullet...\bullet q_n}\\
    0 & 1 & \cdots & s^w_{\bullet q_2\bullet...\bullet q_n}\\
    \vdots & \vdots & \ddots & \vdots\\
    0 & 0 & \cdots & s^w_{\bullet q_n}\\
    0 & 0 & \cdots & 1
\end{bmatrix}= 
\begin{bmatrix}
    1 & f^w_{q_1} & \cdots & f^w_{q_1\bullet...\bullet q_n}\\
    0 & 1 & \cdots & f^w_{q_2\bullet...\bullet q_n}\\
    \vdots & \vdots & \ddots & \vdots\\
    0 & 0 & \cdots & f^w_{q_n}\\
    0 & 0 & \cdots & 1
\end{bmatrix}
$$
where the second equality follows from the definition of $s$, $e$ and $f$. The special minors of the Parikh sequence matrices induced by $q_1\bullet \cdots \bullet q_n$ correspond exactly to the original Parikh matrices induced by $q_1\cdots q_n$, whenever each factor $q_i$ is a single letter. We start by recalling the aforementioned result about Parikh matrices.

\begin{comment}
    
\begin{definition}[The alternate Parikh matrix]
    \cite{MateescuSSY01}Parikh matrix of $w$ is $\Psi_{M_k}(w)=(m_{i,j})_{i\leq i,j\leq k+1}$, the alternate Parikh matrix of $w$ denoted $\overline{\Psi}_{M_k}(w)$ is the matrix $(m'_{i,j})_{i\leq i,j\leq k+1}$, where $m'_{i,j}=(-1)^{i+j}m_{i,j}$ for all $1\leq i,j\leq k+1$.
\end{definition}
\begin{theorem}[\cite{MateescuSSY01}, Theorem 3.2]\label{inverse mirror}
    Let $\Sigma$ and $w\in\Sigma^*$. Then
    $$[\Psi_{M_k}(w)]^{-1}=\overline{\Psi}_{M_k}(mi(w))$$
    $mi(w)$ means the mirror of word $w$.
\end{theorem}

\begin{theorem}[\cite{MateescuSY04}, Corollary 1]
    Let $A$ be the Parikh matrix associated to a word $w$ over an alphabet with $k$ letters. Then every minor of order $k$ of $A$ assumes a nonnegative integer value.
\end{theorem}

\end{comment}

\begin{theorem}[\cite{MateescuSY04}, Theorem 6]\label{minor}
    The value of each minor of an arbitrary Parikh matrix is a nonnegative integer.
\end{theorem}

The proof of the theorem above relies on inverses of Parikh matrices and is not easily reproducible for our more general matrices, so we choose a different argument, which maps those special minors of a Parikh sequence matrix induced by $q_1\bullet \cdots \bullet q_n$ to Parikh matrices defined over an alphabet with $n$ letters. Note that the reduction is not trivially possible by replacing each $q_i$ with a letter $a_i$, as the following example demonstrates.

\begin{example}
Consider the generalized subsequence $q_1\bullet q_2\bullet q_1$ with $q_1=a$, $q_2=aba$ and the matrix tracking its occurrences in the word $aba$:
\[
\Xi_{a\bullet aba\bullet a}(aba)=
\begin{adjustbox}{max width=\linewidth}
$
\begin{bmatrix}
  I & \begin{bmatrix}
  e^{w}_{a\bullet} & e^{w}_{a\bullet a} & e^{w}_{a\bullet ab} & e^{w}_{a\bullet aba\bullet} \\
  0 & e^{w}_{a} & e^{w}_{ab} & e^{w}_{aba\bullet} \\
  0 & 0 & e^{w}_{b} & e^{w}_{ba\bullet} \\
  0 & 0 & 0 & e^{w}_{a\bullet} \\
\end{bmatrix} & \begin{bmatrix}
  f^{w}_{a\bullet a} & f^{w}_{a\bullet ab} & f^{w}_{a\bullet aba} & f^{w}_{a\bullet aba\bullet a} \\
  0 & f^{w}_{ab} & f^{w}_{aba} & f^{w}_{aba\bullet a} \\
  0 & 0 & f^{w}_{ba} & f^{w}_{ba\bullet a} \\
  0 & 0 & 0 & f^{w}_{a\bullet a} \\
\end{bmatrix}\\
  O & \begin{bmatrix}
 1 & e^{w}_{\bullet a} & e^{w}_{\bullet ab} & f^{w}_{ aba} \\ %c matrix
 0 & c^{w}_\varepsilon & c^{w}_{b} & s^{w}_{ba\bullet}\\
 0 & 0 & c^{w}_\varepsilon & s^{w}_{a\bullet}\\
 0 & 0 & 0 & 1\\
 \end{bmatrix} & \begin{bmatrix}
  s^{w}_{\bullet a} & s^{w}_{\bullet ab} & s^{w}_{\bullet aba\bullet} & s^{w}_{\bullet aba\bullet a} \\
  0 & s^{w}_{b} & s^{w}_{ba\bullet} & s^{w}_{ba\bullet a} \\
  0 & 0 & s^{w}_{a\bullet} & s^{w}_{a\bullet a} \\
  0 & 0 & 0 & s^{w}_{\bullet a} \\
\end{bmatrix}\\
  O & O & I
\end{bmatrix}$
\end{adjustbox}
\]
$$=\begin{bmatrix}
  I & \begin{bmatrix}
  2 & 1 & 0 & 0 \\
  0 & 1 & 0 & 1 \\
  0 & 0 & 0 & 1 \\
  0 & 0 & 0 & 2 \\
\end{bmatrix} & \begin{bmatrix}
  1 & 0 & 0 & 0 \\
  0 & 1 & 1 & 0 \\
  0 & 0 & 1 & 0 \\
  0 & 0 & 0 & 1 \\
\end{bmatrix}\\
O & \begin{bmatrix}
 1 & 2 & 0 & 1\\
 0 & 0 & 0 & 0 \\
 0 & 0 & 0 & 2 \\
 0 & 0 & 0 & 1 \\
 \end{bmatrix} & \begin{bmatrix}
  2 & 1 & 1 & 0 \\
  0 & 0 & 0 & 0 \\
  0 & 0 & 2 & 1 \\
  0 & 0 & 0 & 2 \\
\end{bmatrix}\\
  O & O & I
\end{bmatrix}$$

where the special minor is %corresponding to occurrences of $q_1$, $q_2$, $q_1\bullet q_2$, $q_2\bullet q_1$ and $q_1\bullet q_2\bullet q_1$ is 
\[
\begin{bmatrix}
    1 & e^w_{a\bullet} & e^w_{a\bullet aba\bullet} & f^{w}_{a\bullet aba\bullet a}\\
    0 & 1 & f^w_{aba} & s^w_{\bullet aba\bullet a}\\
    0 & 0 & 1 & s^w_{\bullet a}\\
    0 & 0 & 0 & 1
\end{bmatrix}
=
\begin{bmatrix}
    1 & f^w_{a} & f^w_{a\bullet aba} & f^{w}_{a\bullet aba\bullet a}\\
    0 & 1 & f^w_{aba} & f^w_{aba\bullet a}\\
    0 & 0 & 1 & f^w_{a}\\
    0 & 0 & 0 & 1
\end{bmatrix}
=
\begin{bmatrix}
    1 & 2 & 0 & 0\\
    0 & 1 & 1 & 0\\
    0 & 0 & 1 & 2\\
    0 & 0 & 0 & 1
\end{bmatrix}.
\]

If we simply mapped $q_1$ to a letter $c$ and $q_2$ to a letter $d$, then we get $cdc$ as the word inducing the extended Parikh matrix mapping, but the matrix above is not equal to the extended Parikh matrix induced by the word $cdc$ of any word $w\in \{c,d\}^*$, because any word $w$ such that ${w\choose c}=2$ and ${w\choose d}=1$ will have either $cd$ or $dc$ occurring in it as a subword, that is ${w\choose cd}>0$ or ${w\choose dc}>0$, which is in contradiction with entries $(1,3)$ and $(2,4)$.

\end{example}

\begin{theorem}\label{proof minor non integer}
Let $M$ be the special minor of $\Xi_{Q}(w)$ for an arbitrary generalized subsequence $Q=q_1\bullet q_2\bullet ... \bullet q_n$, and an arbitrary word $w$. Then, the determinant of every minor of $M$ is a nonnegative integer.
\end{theorem}
\begin{proof}
The proof proceeds by constructing a word $w'\in \{a_1,\dots,a_n\}^*$ such that $\Psi_{a_1\cdots a_n}(w')=M$. From there, we can apply Theorem~\ref{minor} and get that each minor of $M$ has nonegative determinant. 

\begin{comment}
   
Build a mapping from matrix $F'$ on word $w'$ to a triangle matrix on word $w$,\\ $F'(w') \overset{\delta}{\rightarrow}  \Psi_{M_{n+1}}(w)$:
   $\left\{\begin{matrix} 
  \delta(0)=0,\delta(1)=1  \\
  \delta(q_i)=a_i \\
  \delta(f^{w'}_{q_i\bullet...\bullet q_j)}= |w|_{a_i...a_j}
\end{matrix}\right.$
$a_i,...,a_j\in \Sigma$, word $w \in \Sigma^n$, then\\ $\delta(F')=\begin{bmatrix}
    1 & |w|_{a_1} & \cdots & |w|_{a_1...a_n}\\
    0 & 1 & \cdots & |w|_{a_2...a_n}\\
    \vdots & \vdots & \ddots & \vdots\\
    0 & 0 & \cdots & |w|_{a_n}\\
    0 & 0 & \cdots & 1
\end{bmatrix}$. Next we show how to generate such word $w$. \\
\end{comment}

First we define the order $\{q_1<q_2<...<q_n\}$ on the factors of $Q$, such that if $q_i$ is a prefix of $q_j$, then $q_i<q_j$. For factors $q_i,q_j$ incomparable by the prefix ordering, the relation between them can be set arbitrarily. In the word $w$, of length $k=|w|$, we find all occurrences of the factors $q_1,...,q_n$. Suppose that starting at an arbitrary position $i$ in $w$, we have occurrences of factors $q_{i_1} < q_{i_2} < \cdots < q_{i_s}$. We define the word $u_i = a_{i_s}\cdots a_{i_2}a_{i_1}$. Now we concatenate all the $u_i$ to get the word $w'=\prod_{i=1}^k u_k$.
%get a sequence of factors $p_1,p_2,...,p_t$, where $p_i\in \{q_1,q_2,...,q_n\}$ and $t=f^{w'}_{q_1}+f^{w'}_{q_2}+ \dots +f^{w'}_{q_n}$. Then we do operations on sequence $p_1,p_2,...,p_t$ below,\\

For each $\ell$ running from $n$ to $2$, for each $i$ running from $|w|_{q_\ell}$ to $1$ and for each $j$ running from $|w|_{q_{\ell-1}}$ to $1$, if all of the following conditions apply
\begin{enumerate}
    \item the $i$th occurrence of $q_\ell$ in $w$ overlaps with the $j$th occurrence of $q_{\ell-1}$
    \item the $i$th occurrence of $a_\ell$ is at position $i'$ in $w'$, the $j$th occurrence of $a_{\ell-1}$ is at position $j'$ in $w'$, and $j'<i'$
\end{enumerate} 
then we exchange the positions of the latter two in the word $w'$, that is, we set $w'[i']=a_{\ell-1}$ and $w'[j']=a_\ell$. We run this process iteratively, as long as in the previous iteration there were any changes made, similarly to the bubble-sort algorithm.

The definition of $u_i$, moving $a_p$ in front of $a_r$ when they represent $q_r<q_p$ starting at the same position is done because the $\Psi_{a_1\cdots a_n}$ mapping counts subwords of the form $a_ia_{i+1}$, but does not count $a_{i+1}a_i$, which is in line with a pair $q_i$ and $q_{i+1}$ starting at the same position not contributing towards $f^w_{q_i\bullet q_{i+1}}$. The other rearrangements were necessary, so that if a pair of occurrences of $q_i$ and $q_{i+1}$ are overlapping, the corresponding occurrences of $a_i$ and $a_{i+1}$ will not be counted as an occurrence of subword $a_i a_{i+1}$. These two rules make sure that we do not introduce `extra' tracked occurrences during the construction of $w'$.

The word $w'$ obtained through this process satisfies the conditions
\begin{enumerate}
\item[C1]  
    For any $i$ with $1 \leq i \leq n$, we have $|w'|_{a_i}=f^{w}_{q_i}$.
\item[C2]  
    For any $i,k$, with $1 \leq i < n$ and $k\in [1,|w|_{q_{i+1}}]$, the number of occurrences of $a_i$ before the $k$th occurrence of $a_{i+1}$ in $w'$ is equal to the number of occurrences of $q_i$ that end before the start of the  $k$th occurrence of $q_{i+1}$ in $w$. 
\end{enumerate}

Condition C1. holds, because we map each occurrence of the factor $q_i$ in word $w$ to an occurrence of $a_i$ in $v$, one by one. 

As for condition C2, each individual swap of $a_i$ and $a_{i+1}$ only happens if the corresponding occurrences of $q_i$ and $q_{i+1}$ overlapped, not reducing the number of $a_ia_{i+1}$ occurrences in $w'$ below the number of occurrences of $q_i\bullet q_{i+1}$ in $w$. Such swaps also do not affect the number of occurrences $a_{i-1}a_i$ or $a_{i+1}a_{i+2}$ except in the cases when they would have to be swapped later anyway.

Finally, each occurrence of $q_i\bullet \cdots \bullet q_j$ in $w$ can be uniquely identified by indices $l_i,...,l_j$ where $l_k$ denotes that this occurrence of $q_i\bullet \cdots \bullet q_j$ uses the $l_k-th$ occurrence of $q_k$, for each $k\in [i,j]$, and we know that the $l_k$-th occurrence of $q_k$ does not overlap with the $l_{k+1}$-th occurrence of $q_{k+1}$, for any $k\in [i,j-1]$. This means that there exists an occurrence of the subword $a_i\cdots a_j$ in $w'$ using the $l_k$-th occurrence of $a_k$ in $w'$, for each $k\in [i,j]$, and vice versa. This means that for each $i,j$ with $1\leq i\leq j\leq n$ we have ${w'\choose a_i\cdots a_j}=|w|_{q_i\bullet \cdots \bullet q_j}$, so $\Psi_{a_1\cdots a_n}(w')=M$. From here, by Theorem~\ref{minor} we can conclude our statement, that each minor of $M$ has nonnegative determinant.

%Similarly, each occurrence of $a_i \cdots a_j$ in $w'$ can be uniquely identified by $l_i,\dots,l_j$ where $l_k$ denotes that this occurrence of $a_i\cdots a_k\cdots  a_j$ 
\qed\end{proof}

\section{Generalized subword histories}

In this section we will generalize the concept of subword histories and the main tool used in their analysis, that each subword history is equivalent to a linear one, by a method of eliminating products. First let us recall the notion of subword history defined by Mateescu et al.
\begin{definition}[\cite{MateescuSY04}]
Consider an alphabet $\Sigma$ and a word $w \in \Sigma^*$. A subword history over $\Sigma$ and its value in $w$ are defined recursively as follows. \begin{itemize}
\item Every $u\in \Sigma^*$ is a subword history over $\Sigma$, referred to as a monomial, and its value in $w$ equals ${w\choose u}$.
\item Assume ${SH}_1$ and ${SH}_2$ are subword histories with values $\alpha_1$ and $\alpha_2$, respectively. Then $-({SH}_1)$, $({SH}_1)+({SH}_2)$ and $({SH}_1)\times({SH}_2)$ are subword histories and their values are $-\alpha_1$, $\alpha_1+\alpha_2$ and $\alpha_1\alpha_2$ separately.
\end{itemize}
\end{definition}
In~\cite{MateescuSY04} it is proved that every subword history is equivalent to a linear subword history, i.e., one that does not use $\times$, only $+$, $-$ and monomials. Moreover, given a polynomial $p(u_1,...,u_n)$, a linear polynomial representing an equivalent subword history can be effectively constructed. We will show that subword histories defined through generalized subsequences hold the same properties.
\begin{definition}
Consider an alphabet $\Sigma$ and a word $w \in \Sigma^*$. A monomial generalized subword history over $\Sigma$ and its value in $w$ is defined as a generalized subsequence $Q=q_1\bullet q_2\bullet...\bullet q_n$, with $q_i \in \Sigma^*$, for all $i\in [1,n]$, and its value in $w$ equals $f^w_Q$.
Assume ${GSH}_1$ and ${GSH}_2$ are generalized subword histories with values $\alpha_1$ and $\alpha_2$, respectively. Then $-({GSH}_1)$, $({GSH}_1)+({GSH}_2)$ and $({GSH}_1)\times({GSH}_2)$ are also generalized subword histories, with values $-\alpha_1$, $\alpha_1+\alpha_2$ and $\alpha_1\alpha_2$, respectively. 
\end{definition}

For convenience, we also define the generalized subword histories $\emptyset$, whose value is $0$ in any $w$, $\varepsilon$ whose value is $1$ in any $w$ and define the following identities for all generalized subword histories $Q$ in any word $w$:
\[
\begin{array}{c} 
    \emptyset \bullet Q= Q\bullet \emptyset=\emptyset \\ 
    \varepsilon \bullet Q= Q\bullet \varepsilon=Q\\
    \emptyset + Q = Q + \emptyset = Q.
\end{array}
\] 
Furthermore, $\times$ and $+$ are associative, $+$ is commutative, and $\times$ is left and right distributive over $+$. For the sake of more succinct expression of the definitions below, we also let $\bullet$ be distributive over $+$.

Two generalized subword histories (GSH) are called \emph{equivalent} if their values are the same in any word $w$. A generalized subword history is \emph{linear} if it is obtained from monomials without using the operation $\times$.

For two monomial GSH $Q=q_1\bullet q_2\bullet ...\bullet q_x$ and $P=p_1\bullet p_2\bullet ...\bullet p_y$, we define the ground shuffle of $P$ and $Q$, denoted $P\shuffle Q$ as the set of all sequences $u_0\bullet q_1\bullet u_1...u_{x-1}\bullet q_x\bullet u_x$, where $u_0\bullet \cdots \bullet u_x=p_1\bullet \cdots \bullet p_y$ and any $u_i$ may contain several terms $p_j$ or be empty.

\begin{example}
    For $Q_1=ab\bullet c$, $Q_2=d$, their ground shuffle is $\sum_{x\in Q_1\dsqcup Q_2}x=ab\bullet c\bullet d+ab\bullet d \bullet c+d\bullet ab \bullet c$
\end{example}
    Two GSH $Q_1$ and $Q_2$ are disjoint if they can be defined over disjoint alphabets.

\begin{proposition}\label{lem:disjoint}
    Assume that $Q_1$ and $Q_2$ are disjoint GSH. Then, $Q_1\times Q_2$ and the ground shuffle $\sum_{x\in Q_1\dsqcup Q_2}x$ are equivalent.
\end{proposition}
\begin{proof}
    The argument is straightforward: each pair of occurrences of $Q_1$ and $Q_2$ defines an occurrence of a term in their ground shuffle, since factors of $Q_1$ cannot overlap with factors of $Q_2$, and the other direction is immediate.
    %Consider word $w$ on an arbitrary alphabet $\Sigma$, $Q_1=q^1_1\bullet... \bullet q^1_m$ and $Q_2=q^2_1\bullet... \bullet q^2_n$. If any $q^1_i$ and $q^2_j \notin \Sigma^*$, then the value $Q_1$ or $Q_2$ in $w$ and the value of the ground shuffle in $w$ are 0. Next, assume that $w$ contains the sequences $Q_1$ and $Q_2$. Consider a sequence $x\in Q_1\dsqcup Q_2$ such that $x$ has occurrences of $Q_1$ and $Q_2$ in $w$. Every two factors from two sequences $q^1_i$ and $q^2_j$, either one is not prefix or suffix of another one, so we can confirm which factors in $x$ come from $Q_1$ or $Q_2$. Therefore, the set of occurrences of $x$ in $w$ determines a set of pairs of occurrences of $Q_1$ and $Q_2$ in $w$.
    %$$GSH(w,Q_1\times Q_2)=GSH(w,\sum_{x\in Q_1\dsqcup Q_2}x)$$
\qed\end{proof}
If GSH $Q_1$ and $Q_2$ are not disjoint then we get
$$GSH(w,Q_1\times Q_2)\geq GSH(w,\sum_{x\in Q_1\dsqcup Q_2}x)$$
because each occurrence of a term in $Q_1\shuffle Q_2$ trivially defines a pair of occurrences of $Q_1$ and $Q_2$, but some occurrence of $Q_1$ might overlap with an occurrence of $Q_2$, so that pair would not correspond to any ground shuffle term occurrence.

\subsection{Reduction to linear generalized subword histories}
Now we move on to show how to construct a linear GSH equivalent to a given GSH that uses the $\times$ operation. Let us denote by $\mathbf{GSH}_\Sigma$ the set of all GSH over $\Sigma$. For $\Sigma$, we consider the primed version $\Sigma'=\{a'|a\in \Sigma\}$. 
%Let $g:{\Sigma}^*\to{\Sigma'}^*$ be the morphism defined by $g(a)=a'$. Also, 
Let $h:(\Sigma\cup\Sigma')^*\rightarrow \Sigma^*$ be the morphism defined by $h(a)=h(a')=a$. By a slight abuse of the notation we extend it to the domain of monomial GSH as $h(Q)=h(q_1)\bullet \cdots \bullet h(q_n)$ for any $Q=q_1\bullet \cdots \bullet q_n$, where $q_i \in (\Sigma\cup\Sigma')^*$, and then further as $h(a_1Q_1+a_2Q_2)= a_1h(Q_1)+a_2h(Q_2)$, where $a_1,a_2\in \mathbb{Z}$ and $Q_1,Q_2\in \mathbf{GSH}_{(\Sigma\cup\Sigma')}$. Let the morphism $g:\Sigma^*\rightarrow \Sigma'^*$ be defined as $g(a)=a'$ for all $a\in \Sigma$ and extended to GSH similarly to $h$.

For the generalized subsequences $Q_1=p_1\bullet \cdots \bullet p_m$ and $Q_2=q_1\bullet \cdots \bullet q_n$, both occurring in $w$, a pair of occurrences of $Q_1$ and $Q_2$ is \emph{interleaved} if for each $i\in [1,m-1]$ there is a $j\in [1,n]$ such that $q_j$ overlaps with both $q_i$ and $q_{i+1}$ and for each $i\in [1,n-1]$ there is a $j\in [1,m]$ such that $p_j$ overlaps with both $q_i$ and $q_{i+1}$.

%For two sequences $Q_1=q^1_1\bullet... \bullet q^1_m$ and $Q_2=q^2_1\bullet... \bullet q^2_n$, we consider their ground shuffle on every factor:
%$$G(Q_1,Q_2)=Q_1\dsqcup g(Q_2)$$

Now we introduce a function to reduce the $\bullet$ of two monomial GSH defined over $\Sigma$ and $\Sigma'$, respectively. Suppose that the GSH given by the generalized subsequence $Q_1=p_1\bullet \cdots \bullet p_i$ is over $\Sigma$ and $Q_2=q_1\bullet \cdots \bullet q_j$ is over $\Sigma'$. We define the function $red$, such that $red(p_1\bullet...\bullet p_i \bullet q_1\bullet... \bullet q_j)=a_1v_1+a_2v_2+\dots +a_kv_k$, where for all $d\in [1,k]$, we have $a_d\in \mathbb{N}$ and $v_d\in \Sigma^+$ are factors satisfying the conditions:
        \begin{enumerate} 
            %\item $|v_d|_{p_1\bullet \cdots \bullet p_i} \geq 1$, $|v_d|_{h(q_1\bullet \cdots \bullet q_j)} \geq 1$, and $|v_d|_{p_1\bullet...\bullet p_i \bullet h(q_1\bullet... \bullet q_j)} =0$. 
            \item $v_d$ contains interleaved occurrences of $Q_1$ and $h(Q_2)$ starting at the first position of $v_d$ and ending at the last position of $v_d$
            %with prefix $p_1$ or $h(q_1)$, and ends with suffix $p_i$ or $h(q_j)$. In $v_d$, the prefix $p_1$ overlaps with an occurrence of $h(q_1)$, or vice versa, and suffix $p_i$ overlaps with an occurrence of $h(q_j)$, or vice versa.
            \item $a_d$ is the number of distinct interleaved occurrences of $Q_1$ and $h(Q_2)$ in $v_d$ starting at the first position of $v_d$ and ending at the last position of $v_d$.
            %occurrence pairs of $p_1\bullet \cdots \bullet p_i$ and $h(q_1\bullet \cdots \bullet q_j)$ in $v_d$ such that $p_1$ (or $h(q_1)$) is a prefix overlapping with an occurrence of $h(q_1)$ (of $p_1$) and $p_i$ (or $h(q_j)$) is a suffix overlapping with an occurrence of $h(q_j)$ (of $p_i$).
        \end{enumerate}
        If such $v_d$ exist, then we say that $Q_1\bullet Q_2$ is reducible. If $v_d$ does not exist, then $red(Q_1\bullet Q_2)=\emptyset$. Note that $|v_d|<|p_1|+\cdots +|p_i|+|q_1|+\cdots+|q_j|$, due to the conditions on interleaved occurrences stretching across $v_d$, therefore $|v_d|_{p_1\bullet...\bullet p_i \bullet h(q_1\bullet... \bullet q_j)} =0$.

\begin{example}
    Assume that $Q_1= abc\bullet c$ and $ Q_2=a'\bullet c'$. Then $red(Q_1\bullet Q_2)=\emptyset$. However, if $Q_3=abc\bullet de$, $Q_4=a'\bullet c'd'$, then $red(Q_3\bullet Q_4)=a_1v_1$ with $a_1=1$ and $v_1=abcde$.
\end{example}

Now we can construct for the general case of GSH over $\Sigma\cup \Sigma'$ the reduction function $R:\mathbf{GSH}_{(\Sigma\cup \Sigma')^*}\rightarrow \mathbf{GSH}_{\Sigma^*}$ inductively as follows. 
\begin{enumerate}
    \item For GSH $P$ over $\Sigma$ and $Q$ over $\Sigma'$, let $R(P)=P$ and $R(Q)=h(Q)$.
    \item For monomial GSH $P$ over $\Sigma$ and $Q$ over $\Sigma'$ let $R(P\bullet Q)=P\bullet h(Q) + red(P\bullet Q)$.
    \item For monomial GSH $P_1,\dots, P_m$ over $\Sigma$ and $Q_1,\dots, Q_m$ over $\Sigma'$, with $m>1$:
    \[
    R(P_1\bullet Q_1\bullet P_2\bullet Q_2 \bullet \cdots \bullet P_m\bullet Q_m)= R(P_1\bullet Q_1)\bullet R(P_2\bullet Q_2\bullet \cdots \bullet P_m\bullet Q_m).
    \]
    \item For GSH $Q$ over $\Sigma\cup \Sigma'$ and $z\in \mathbb{Z}$, let $R(zQ)=zR(Q)$.
\end{enumerate}

\begin{theorem}
    For arbitrary monomial GSH $Q_1$ and $Q_2$, the GSH $Q_1\times Q_2$ is equivalent to the linear GSH $$\sum_{Q\in Q_1\shuffle g(Q_2)} R(Q).$$
\end{theorem}
\begin{proof}
    There is a one-to-one correspondence between pairs of occurrences of $Q_1,Q_2$ and the set of occurrences of (1) terms of the ground shuffle of $Q_1$ and $Q_2$ and (2) reduced terms of the ground shuffle, where the $red$ function applied to some factor. The former is the case when the pair of occurrences of $Q_1$ and $Q_2$ has no interleaved parts, and the latter otherwise. The correspondence is straightforward from the construction of the reduction functions. The value of $\sum_{Q\in Q_1\shuffle g(Q_2)} R(Q)$ is clearly greater or equal to the value of $Q_1\times Q_2$, due to rule 2. in the definition of $R$. Conversely, each term in $\sum_{Q\in Q_1\shuffle g(Q_2)} R(Q)$ that does not occur in the ground shuffle, is obtained as $P_1\bullet red(Q'_1\bullet Q'_2)\bullet P_2$ for some ground shuffle term containing the subsequence $Q'_1\bullet Q'_2$ and corresponds to counting occurrences pairs of $Q_1$ and $Q_2$ where the $Q'_1$ part of $Q_1$ is interleaved with the $Q'_2$ part of $Q_2$.

\qed\end{proof}

By applying the laws of distributivity of $\times$ over $+$ and the reduction above, we can obtain an equivalent linear GSH from any starting GSH. An important consequence of the ability to linearize GSH is that, by a simple proof just like in the case of the original subword histories, two linear GSH are equivalent if and only if they are identical up to a reordering of their terms. This means that equivalence of GSH is decidable by a simple, albeit high time complexity algorithm.

\section{Conclusion}
We proposed a generalization of Parikh matrices that are able to track factors and subsequences with gaps. For these new mappings, We also obtained the generalized versions of the most fundamental properties applied in the study of Parikh matrices: that they are homomorphisms and that significant minors of the matrices have nonnegative determinant. The latter property gave rise to a well-investigated direction of research into equations and inequalities involving subword histories. To be able to continue the study of subsequences with gaps in that direction, we also showed that subword histories defined with those generalized subsequences can also be linearized, just as was the case with the subword histories based on scattered subword occurrences. One open road ahead that we think is worth investigating, is to describe the class of languages that can be defined through generalized subword history conditions.

Another meaningful way to expand on this research is to consider the `ultimate' generalizations of subwords: subsequences with gap constraints provided as length (the maximum allowed distance individually defined between each consecutive pair of letters of the subsequence when matched in the containing word) or even more generally subsequences with regular gap constraints (for exciting recent results, see~\cite{ManeaRS24}), that is, when matching the subsequence, the gap between two consecutive letters has to be a word in a given regular language. The formalism we developed seems flexible enough to handle length constraints by minor modification, but we do not yet know whether the regular constraints are implementable through such matrix mappings.

%\section{References}
%\renewcommand{\refname}{}

\bibliographystyle{splncs04}
\bibliography{Parikh_factor_matrix}

\end{document}